\documentclass[12pt]{article}
\usepackage[margin=1in]{geometry}
\usepackage{color}
\usepackage{graphicx,enumerate}
\usepackage{amsfonts}
\usepackage{amsmath,pslatex,algorithmic}
\usepackage[ruled, vlined, linesnumbered, noresetcount]{algorithm2e}
\usepackage{url}

\usepackage{verbatim}

\usepackage{scrextend}
\usepackage{tikz}
\usepackage{graphics}
\usepackage{epsfig}
\usetikzlibrary{through,arrows}
\tikzset{
  dot node/.style={
    shape=circle,
    fill=white,
    draw,
    inner sep=+0pt,
    minimum size=+7mm
  },
  dotdot node/.style 2 args={
    dot node,
    label={[shape=circle,fill=black,outer sep=+0pt,inner sep=+0pt,minimum size=+3mm,name=ddd-#1,#2]center:}
  },
  arc style/.style={
    |<->|,
    shorten >=+-.5\pgflinewidth,
    shorten <=+-.5\pgflinewidth,
  }
}
\newtheorem{example}{Example}
\newtheorem{definition}{Definition}

\newtheorem{lemma}{Lemma}
\newtheorem{obs}{Observation}
\newtheorem{theorem}{Theorem}

\newtheorem{proposition}{Proposition}

\newenvironment{proof}{\noindent{\bf Proof.}}{\hfill \qed \vskip 5pt}
\def\qed{\hfill\rule{2mm}{2mm}}

\date{}

\newcommand{\G}{{\mathcal G}}

\newcommand{\ML}{{\sf Min-Links }}

\def\proof{\noindent{\bf Proof. $\ $}}
\def\qed{\hfill $\Box$}

\def\s{{\sigma}}
\def\d{{\delta}}
\def\bs{{\bf s}}
\def\cs{s}

\newtheorem{prop}{Proposition}

\newcommand{\red}[1]{#1}

\newcommand{\remove}[1]{}

\def\bs{{\bf s}}
\def\cs{s}

\newcommand{\Active}{{{ I}}}

\def\G{{{\cal G}}}
\def\vj{{v_j}}

\def\D{{\rm  B}}

\begin{document}


\thispagestyle{empty}

\title{\bf Whom to befriend to influence people
} %
\author{Gennaro Cordasco\thanks{Department of Psychology, Universit\`a degli Studi della Campania ``Luigi Vanvitelli'', Italy}
\and
Luisa Gargano\thanks{Department of Computer Science, Universit\`a degli Studi di Salerno, Italy}
\and
Manuel Lafond\thanks{
Department of Computer Science and Operational Research, Universit\'e de Montr\'eal, Montr\'eal, Canada}
 \and 
Lata Narayanan\thanks{Department of Computer Science and Software Engineering, Concordia University, Montr\'eal, Canada}
\and 
Adele~A.~Rescigno\footnotemark[2]
\and 
Ugo Vaccaro\footnotemark[2]
 \and 
Kangkang Wu \footnotemark[4]
}                      

\maketitle


\begin{abstract}
Alice wants to join a new social network, and influence its members to adopt a new product 
or idea. Each person $v$ in the network has a certain threshold $t(v)$ for {\em activation}, 
i.e adoption of the product or idea. If  $v$ has at least $t(v)$ activated neighbors,  then 
$v$ will also become activated. If Alice wants to activate the entire social network, whom should she befriend? 
More generally,  we study the problem of finding the minimum number of links that  a set of  external 
influencers should form to people  in the network, in order to activate the entire social network.  
This  {\em Minimum Links}  Problem  has applications in viral marketing and the study of epidemics.  
Its solution can be quite different from the related and widely studied Target Set Selection problem. 
We prove that the Minimum Links problem cannot be approximated to within a ratio 
of  $O(2^{\log^{1-\epsilon} n})$,  for any fixed $\epsilon>0$, 
unless $NP\subseteq DTIME(n^{polylog(n)})$, where $n$ is the number of nodes
in the network. On the positive side, we give linear time algorithms to solve the problem for trees, 
cycles, and cliques, for any given set of  external influencers, and give precise bounds on the number of links needed. 
For general graphs,  we design  a polynomial time algorithm to  compute  size-efficient link sets
that can activate the entire graph.
\end{abstract}
%
\newcommand{\ac}{\sim}
\newcommand{\extinf}{\mu}

%
%


\section{Introduction}

The increasing popularity and proliferation of large online social networks, together with the availability of enormous amounts of data about customer bases,  has contributed to the rise of {\em viral marketing }  as an effective 
strategy in promoting 
new products or ideas. This strategy relies on the insight that once a certain 
fraction of a social network adopts a product, a larger cascade of further 
adoptions is predictable due to the {\em word-of-mouth}  network effect 
\cite{GLM01,LAH07,BR87}.
Inspired by social networks and viral marketing, Domingos and Richardson 
\cite{DR-01,RD02}  were the first to raise the following
important algorithmic problem in the context 
of social network analysis: 
If a company can turn a subset of customers in a given network into 
early adopters, and the 
goal is to trigger a large cascade of further adoptions, 
which set of customers should they target?

We use the well-known threshold model to study the influence diffusion process in social networks from an algorithmic perspective. The social network is modelled by a node-weighted graph \textbf{$G= (V,E, t)$} with $V(G)$ representing individuals in 
the social network, $E(G)$ denoting the social connections, and $t$ an integer-valued {\em threshold function}. 
Starting with a {\em target set}, that is, a subset $S \subseteq V$ of nodes in the graph, that are {\em activated} by some external incentive, influence propagates deterministically in 
discrete time steps, and {\em activates} nodes. 

For any unactivated 
node $v$, if the number of its activated neighbors at time step $t-1$ is 
at least $t(v)$, then node $v$ will be activated in step $t$. A node once activated 
stays activated. 
It is easy to see that if $S$ is non-empty, then the process terminates after 
at most $|V|-1$ steps. We call the set of nodes that are activated when the 
process terminates as the {\em  activated set}.  
The problem proposed by Domingo and Richardson \cite{DR-01,RD02} can now be 
formulated as  follows:
Given a social network $G= (V, E, t)$, and an integer $k$, find a subset $S
\subseteq V$ of size $k$ so that the resulting activated set is as 
large as possible. In the context of viral marketing, the parameter $k$ corresponds to the budget, and $S$ is a target  
set that maximizes the size of the  activated set. One question of interest is to find the cheapest way to activate the {\em entire network}, when possible. The optimization problem that results  has been called the {\em Target Set Selection Problem}, and has been widely studied (see for eg.  \cite{Chen-09,BHLN11,NNUW}): the goal is to find a minimum-sized set $S \subseteq V$ that activates the entire network (if such a set exists).  In a certain sense, the elements of this minimum target set $S$ are the most influential people in the network; if they are activated, the entire network will eventually be activated.

There are, however,  two hidden flaws in the formulation of the  target set 
problem.  First, the nodes in the target set are assumed to be activated  immediately by external incentives, {\em regardless of their own thresholds of activation}. 
  This is not a realistic assumption; in the context of viral marketing, it is possible, perhaps even  likely, that highly influential nodes have high thresholds, and cannot be activated by external incentives alone. Secondly,  there is no possibility of giving {\em partial} external incentives; indeed the  target set is activated {\em only}  by external incentives, and the remaining nodes {\em only} by the internal network effect. 

In this paper, we address the flaws mentioned above. We study  a related but different problem. Suppose Alice  wants to join a new social network, whom should she befriend if her goal is to influence the entire social network? In other words, to whom should Alice create links, so that she can activate  the entire network? If Alice creates a link to a node $v$, the threshold of $v$ is only effectively reduced by one, and so $v$ in turn is activated only if its threshold is one. The problem can be generalized to any set of $k$ external influencers that wish to collectively "take over" a network. We call our problem the {\em Minimum Links} problem ({\sf Min-Links}). 

The {\sf Min-Links} problem provides a new way to model a viral marketing strategy, which 
addresses the flaws described in the target set problem formulation. 
The links added from the external nodes correspond to the external incentive given to the endpoints of these links. The  nodes that are the endpoints of these new links may not be immediately completely activated, but their thresholds are effectively reduced; this corresponds to their receiving partial incentives. One way of seeing this is that every individual  to whom we link is given a \$10 coupon; for some people this may be enough for them to buy the product, for others, it reduces their resistance to buying it. Individuals with high thresholds cannot be activated only by external incentives. The {\sf Min-Links} problem also has important applications in epidemiology or the spread of epidemics: in the spread of a new disease, where an infected person or a set of infected people arrives from {\em outside} a community, the {\sf Min-Links} problem corresponds to  identifying the smallest set of people such that if the infected external people have contact with this set, the entire community could potentially be infected.

Observe that the solution to the {\sf Min-Links} problem can be quite different from the solution to the  Target Set  Selection problem for a given network.  For example, consider a star network, where the leaves all have threshold 1, while the central node has degree $|V|-1$ and has threshold $|V|$. The optimal  target set  is the central node, while the only solution to the {\sf Min-Links} problem with a single influencer is to create links to all nodes in the network. Thus, a solution to the {\sf Min-Links} problem can be arbitrarily larger than one to the  Target Set Selection problem for the same social network. However, any solution to the {\sf Min-Links} problem is clearly also a feasible solution to the  Target Set Selection problem. 

\subsection{\red{Our Results}}

 We prove that  there exists a (gap-preserving) reduction from the classical
Target Set Selection problem to the {\sf Min-Links}  problem.  Using the important 
results by \cite{Chen-09}, this implies that 
the  {\sf Min-Links}  problem, even in presence of  a single external influencer,
 cannot be approximated to within a ratio 
of  $O(2^{\log^{1-\epsilon} n})$,  for any fixed $\epsilon>0$, 
unless $NP\subseteq DTIME(n^{polylog(n)})$, where $n$ is the number of nodes in the graph. 
 In light of this  hardness result, we study the complexity of the problem for  networks that can be represented as trees, cycles, and cliques. 
In each case, we give a necessary and sufficient condition for the feasibility of the {\sf Min-Links} problem, based on the structural properties and an observation of the threshold function.  We then give $O(|V|)$ algorithms to solve  the {\sf Min-Links} problem for all the studied graph topologies.  We also give exact bounds on the number of links needed to activate the entire network for all the above specific topologies,  as a function of the threshold values. 
Finally, we present
a polynomial time algorithm {that}, given an arbitrary   network $G$ and
a number of influencers equal to the maximum node threshold,   computes a 
 ``small''    set of links  sufficient to   activate the whole network. Our  polynomial time algorithm 
always returns a solution    for $G$ of size at most 
$ \sum_{v\in V}  
\frac{t(v)(t(v) +1)}{2(d_G(v)+1)}$, where $d_G(v)$ is the degree of the vertex  $v$.
\remove{
exhibits the following features: 
\begin{enumerate} 
\item For general graphs, 
it always return a solution    for $G$ of size at most 
$ \sum_{v\in V}  
\frac{t(v)(t(v) +1)}{2(d(v)+1)}$.
\item For trees and complete graphs our algorithm always returns an  \emph{optimal} pervading link set.
\end{enumerate}
}

\subsection{Related work}
The problem of identifying the most influential nodes in a social network has received a tremendous amount of attention \cite{GBLV13,HJBC14,LBGL13,GLL11,CWY09,GBL11,BBC14,FGH12}. The algorithmic question of choosing the target set of size $k$ that activates the most number of nodes in the context of viral marketing was first posed by Domingos and Richardson \cite{DR-01}. Kempe \emph{et al.} \cite{KKT03} started the study of this problem as a discrete optimization problem, and studied it in both the probabilistic independent cascade model and the threshold model of the influence diffusion process. They showed the NP-hardness of the problem in both models, and showed that a natural greedy strategy has a $(1 - 1/e - \epsilon)$-approximation guarantee in both models; these results were generalized to a more general cascade model in \cite{KKT05}. 

In the  Target Set Selection problem, the size of the target set is not specified in advance, but the goal is to activate the entire network. 
Namely, given a graph $G$ and fixed arbitrary thresholds $t(v)$, $\forall v\in V$, find
a target set of minimum size that eventually activates
all (or a fixed fraction of) nodes of $G$.
Chen \cite{Chen-09} proved  a strong   inapproximability result for the 
Target Set Selection problem that makes unlikely the existence
of an  algorithm with  approximation factor better than  $O(2^{\log^{1-\epsilon }|V|})$.
 A polynomial-time algorithm for trees was given in the same paper. 
Chen's inapproximability  result stimulated a series of papers  
(see for instance \cite{ABW-10,BCNS,BHLN11,Centeno12,Chang,Chun2,Chun,Chopin-12,Cic+,Cic14,C-OFKR,Fr+,GHPV13,Gu+,Li+,Mo+,NNUW,Re,W+,Za} and references therein quoted)
that isolated 
many interesting scenarios 
in which the problem and variants thereof become tractable.
Ben-Zwi et al. \cite{BHLN11} generalized Chen's result on trees to show  that target set selection can be solved in $n^{O(w)}$ time where $w$
 is the treewidth of the input graph. The effect of several parameters,  such as diameter and  vertex cover number,
 of the input graph on the complexity of the problem are studied in \cite{NNUW}. The Minimum Target Set has also been studied 
from the point of view of the spread of disease or epidemics. For example, 
in \cite{DR09}, the case when all nodes have a threshold $k$ is studied; the authors showed that the problem is NP-complete for fixed $k \geq 3$.

Maximizing the number of nodes activated within a specified number of rounds has also been 
studied \cite{DZNT14,LP14}. The problem of  dynamos or dynamic monopolies in graphs 
(eg. \cite{Peleg02}) is essentially  the target set problem restricted to the case when every node's threshold is half its degree. 
The recent monograph \cite{CLC} contains an excellent overview of the area.

The paper closest to our work is \cite{DHM14}, in which Demaine {\em et al.} introduce a model to {\em  partially incentivize} nodes to maximize the spread of influence. Our work differs from theirs in several ways. First, they study the maximization of influence given a fixed budget, while we study in a sense the budget (number of links) needed to activate the entire network. Second, they consider thresholds chosen uniformly at random, while we study arbitrary thresholds.  Finally, they allow arbitrary fractional influence to be applied externally on any node, while in our model, every node that receives a link has its threshold reduced by the same amount.

\section{Notation and preliminaries.} \label{notation}
Given a social network represented by an undirected graph $G=(V, E, t)$, we introduce a  set of external nodes $U$ that are assumed to be already activated. We assume that all edges have unit weight; this is 
generally called the {\em uniform weight assumption}, 
and has previously been considered in many papers
\cite{Chen-09,GHPV13,Cic+}.
A {\em link set} for $(G, U)$ 
is a set  $S$ of  links between nodes in $U$ and nodes in $V$, i.e $S \subseteq \{ (u, v) \mid u \in U; v \in V\}$. For a link set $S$, we define  $E(S) = \{v  \in V\mid \exists (u, v) \in S \}$, that is, $E(S)$ is the set of $V$-endpoints of links in $S$. 
For a node $v$, define $s(v)$ to be the number of links in $S$ for which $v$ is an endpoint. 
Since the set of external  nodes $U$ is already activated, observe that adding the link set $S$ to $G$ is equivalent to reducing the threshold of the node $v$ by $s(v)$. In the viral marketing scenario, the link set $S$ represents giving $v$ a partial incentive of $s(v)$  \cite{Sirocco15,DHM14}. 

Given a link set $S$ for a graph $G$, we define  $I(G, S)$ to be the set of 
nodes in $G$ that are eventually activated as a result of  adding the link set $S$, that is, by reducing the threshold of each node  $v \in E(S)$ by $min \{ s(v), t(v) \}$, and then running the influence diffusion process. See Figure~\ref{link-set-example} for an illustration. 
Observe that in the target set formulation, this is the same as the set of nodes activated by using $U$ as the target set in the graph $G'$,  the graph obtained from $G$ by adding the set $U$ to the node set and the set $S$ to the set of edges.

 \begin{figure} 
 \centering
 \includegraphics[height=2.0in]{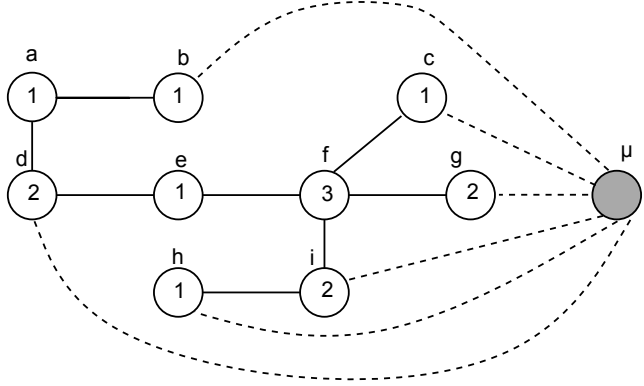}
 \centering
 \caption{Node $\mu$ is the only  external influencer and is assumed to be activated. Links in the link set are shown with dashed edges. The given link set activates the entire network and is an optimal pervading link set.}
 \label{link-set-example}
 \end{figure}

 A link set $S$ such that $I(G, S) = V$, that is, $S$ activates the entire network, is called a {\em pervading link set}.  A pervading  link set of minimum size is called an {\em optimal pervading link set.}  
\begin{definition}
	{\bf Minimum Links ({\sf Min-Links}) problem:} Given a social network $G=(V, E, 
	t)$, where $t$ is the threshold function on $V$, and a set of external nodes $U$, find an  optimal pervading link set for $(G, U)$.
		\end{definition}

For each node $v \in E(S)$, we say we {\em give} $v$ a link, or that $v$ {\em receives} a link. In our algorithms, we express a link  set as a {\em link vector}, 
$\bs = (s(v_1), \ldots , s(v_n))$, where 
$s(v)$, as defined earlier, is an integer representing  the number of links between  external nodes in $U$ and  the vertex $v \in V$. The external influencer-endpoints of these links are understood to be distinct, but otherwise can be chosen arbitrarily within $U$. If activating $X \subseteq V$ activates, directly or indirectly, the set of vertices $Y$, we write
$X \ac Y$ (note that there may be vertices outside $Y$ that $X$ activates).  
We write $x \ac Y$ instead of $\{x\} \ac Y$.  
The minimum cardinality of a link set for  a {\sf Min-Links} instance $G$ is denoted $ML(G)$.

Observe that for some graphs, and some sizes of the external influencer set, a  pervading link set  may not exist. For example, consider a singleton node of threshold greater than 1, and a single influencer. The existence of a feasible solution can be verified in $O(E)$ time by giving $k$ links to every node in $V$, and simulating the influence diffusion process.

\section{NP-hardness}\label{hard}

In this section, we consider the complexity of    the {\sf Min-Links} problem. 
We  prove the following result.
\begin{theorem}
In networks with $n$ nodes 
 {\sf Min-Links} problem  cannot be approximated to within a ratio 
of  $O(2^{\log^{1-\epsilon} n})$ for any fixed $\epsilon>0$, 
unless $NP\subseteq DTIME(n^{polylog(n)})$, even if the network has bounded degree, and all thresholds are at most 2. 
\end{theorem}
\proof
We  construct  a gap-preserving reduction from the   Target Set Selection  (TSS) problem. 
We recall that, given an input graph $G$ and  a threshold function $t$,   the TSS problem asks for a minimum size subset of vertices of $G$ that can activate all the other vertices.
The inapproximability claim of the theorem follows from the inapproximability result of TSS 
proved in \cite{Chen-09}, which holds even for graphs with bounded degree, when the thresholds are at most 2. 

Starting from an input instance  of 
the TSS problem, that is, a bounded-degree graph  $G=(V,E, t)$ with threshold function $t$, such that $t(v) \leq 2$ for all vertices $v\in V$, we build an instance of the Min-Links problem with a single influencer. Define the   graph $G'=(V', E', t')$ as follows:

\begin{itemize}
\item $V'=\bigcup_{v\in V}V'_v$ where $V'_v=\{v',v'',v_1,\ldots v_{t(v)}\}$. 
		 In particular, 
\begin{itemize}
		\item  we replace each  $v\in V$ by the gadget $\Lambda_v$ (cf. Fig. \ref{fig:gadget}) in which the node set is  $V'_v$ and  $v'$ and  $v''$  are connected by the 
		disjoint paths  ($v',v_i,v''$) for $i=1,\ldots, t(v)$; 
	\item the threshold of $v'$ in $G'$ is equal to the threshold  
	$t(v)$ of $v$ in $G$, while each other node in $V'_v$ has threshold equal  to 1.
	\end{itemize}	
\noindent
$\bullet$    $E'=\{ (v',u')\ |\  (v,u)\in E\} \cup \bigcup_{v\in V} \{(v',v_i) , (v_i,v'')\mbox{, for  } i=1,\ldots, t(v)\}$.
\end{itemize}
Summarizing,  $G'$ is constructed in such a way that for each gadget $\Lambda_v$, the node $v'$ plays the role of $v$ and is connected to all the gadgets representing neighbors of $v$ in $G.$ 
Hence, $G$ corresponds to  the  subgraph of $G'$ induced by the set $\{v'\in V'_v |\  v \in V\}.$
It is worth  mentioning that during an influence diffusion process if any node that belongs to a gadget  $\Lambda_v$ is active, 
then all the vertices in $\Lambda_v$ will be activated within the next  $3$ steps.
Moreover, there is only one influencer and the influencer set is $U=\{\mu\}$. Observe that all thresholds in $G'$ are at most $2$, and $G'$ remains of bounded degree. 

We claim that there is a target set $T \subseteq V$ for $G$ of cardinality $|T|=k$ \emph{if and only} if there is a pervading link set  for $(G',U)$ of size $k$.
Assume that $T\subseteq V$ is a target set for $G$, we consider the set of links $S'$, with $|S'|=k$, defined as
$$S'=\left\{(\mu,v'')\ |\ \mbox{$v''$ is the extremal node  in the gadget $\Lambda_v$ and $v \in T$}                                         \right\}.$$
 To see that $S'$ is a pervading link set, we notice that $S' \ac \{u \ | \ u \in V'_v, v\in T\}$ 
within three steps. Consequently, recalling that $T$ is a target set and that  $G$ is 
isomorphic to the  subgraph of $G'$ induced by  $\{v'\in V'_v |\  v \in V\}$, all the vertices  $v \in V'$ will be activated, that is $I(G',S')=V'$.

\medskip 
On the other hand, assume that $S'$ is a pervading link set for $(G',U)$ and   $|S'|=k$, we can easily build a target set 
$$
T=\{v \in V \ | \  \mbox{there exists } w \in V'_v \mbox{ such that } (\mu, w)\in S' \}.   
$$
By construction $|T| \leq |S'|$. We show now that $T$ is a target set for $G$.
To this aim,  for each $v\in V$ we consider two cases 
according to how the node $v'  \in \Lambda_v$ associated with $v$ is activated in $G'$:
\begin{itemize}
\item  
If there exists $w\in V'_v$ such that $(\mu, w)\in S'$ then, by construction $v \in T$.
\item
 Suppose otherwise  that  for each  $w\in V'_v$ it holds that $(\mu, w)\notin S'$.
In order to activate $v'$ (and afterwards 
	any other node in  $\Lambda_v$), there must exist a step  $i$ 
	when at least $t(v)$ of the neighbors of $v'$  in $V'-V'_v$  are active.
\end{itemize}
	Now we recall 
 that $G$ is the subgraph of $G'$ induced by the set $\{v'\in V'_v \ |\  v \in V\}$. Hence,  for each step $i\geq 0$ and for each    $v'$ which is active in $G'$ at step $i$
	(with link set $S'$), 
	we conclude that the corresponding node $v$ must be active in $G$ by step $i$ (with target set $T$).
Consequently any node $v$ will be activated in $G$. 
\qed

\begin{figure}	

\begin{center}
	\includegraphics[height=4.3truecm]{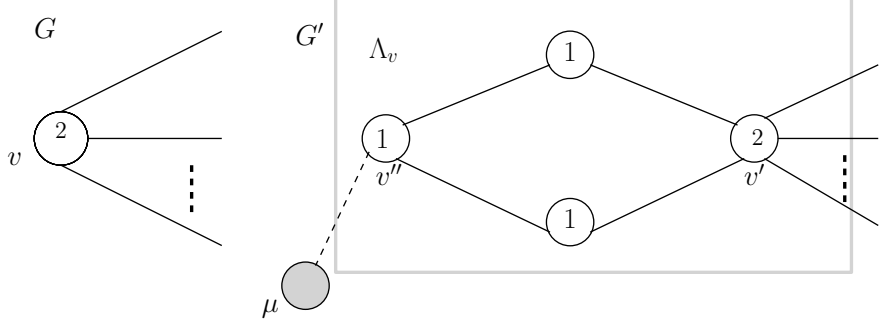}

		\caption{The gadget $\Lambda_v$: (left) a generic node $v \in V$ 
		having degree $d_G(v)$ and threshold $t(v)=2$; (right) the gadget $\Lambda_v$, having $t(v)+2=4$ vertices, associated to $v$.  \label{fig:gadget}}
\end{center}

\end{figure}

\bigskip

In the case of  
very small degree bound, it has been proved in \cite{Sirocco16} that 
 the {\sf Min-Links} problem is NP-hard; in fact, it is almost as hard as {\sf Set-Cover} to approximate,
even if $G$ has degree bounded by $3$ and thresholds bounded by $2$.

\begin{theorem}\cite{Sirocco16}
The decision version of {\sf Min-Links} is NP-complete, even when restricted to instances with maximum degree $3$ and maximum threshold $2$.
Moreover, there exists a constant $\epsilon > 0$ such that the optimization version of {\sf Min-Links}, under the same restrictions, is NP-hard to approximate
within a $\epsilon \ln n$ factor, where $n$ is the number of nodes of the given graph.
\end{theorem}

\section{Algorithms for MinLinks}\label{Sec4}

In this section, we give linear time algorithms to solve the MinLinks problem for trees, cycles, and cliques,  for  any given set of $k$ external influencers (i.e. a single node is able to receive up to $k$ links).
Hereafter, the external influencers are $U = \{\extinf_1, \ldots, \extinf_k\}$ and a solution $S$ for a graph $G$ consists 
in a set of distinct links $(\extinf_i, v)$ where $\extinf_i \in U$ and $v \in V(G)$. 
We start with the following simple observation:

\begin{obs} \label{impossible}
A graph $G$ does not have a pervading link set if it has a node $v$ such that $t(v) > degree(v) + k$,  or if every node has threshold strictly greater than $k$.

\end{obs}

\subsection{Trees}

In contrast to the NP-completeness of the {\sf Min-Links} problem shown in the previous section, we now show that there is a linear time algorithm to solve the problem in trees. We start with a necessary and sufficient condition for a tree $T$ to have a valid pervading link set.

\begin{proposition}
Let $T$ be a tree and let $v$ be a leaf  in $T$. Let $T' = T - \{v\}$ and $T''$ be the same as $T'$ except that the threshold of $w$, the neighbor of $v$ in $T$, is reduced by 1. Then $T$ has a pervading link set  if and only if (a) either  $t(v) \leq k$ and $T''$ has a pervading link set or (b) $t(v) = k + 1$ and  $T'$ has a pervading link set. 

\label{necc-suff-tree}
\end{proposition}

We now prove a critical lemma that shows that for any node $v$ in the tree, there is an optimal solution that gives $\min (t(v), k)$ links  to $v$.

\begin{lemma}\label{lemma:tree}
Let $T$ be a tree with $n$ nodes that has a pervading link set, and let $v$ be a node in $T$. Then there exists an optimal solution for {\sf Min-Links}$(T)$ in which $v$ gets $\min \{t(v), k\}$ links.
\end{lemma}

\begin{proof}
We prove the lemma by induction on the number of nodes $n$ in the tree. Clearly it is true if $n=1$. Suppose $n>1$, and let $S$ be an optimal pervading link set for $T$.  
Moreover, choose $S$ such that $v$ receives a maximum number of links among all optimal solutions.
If $v$ gets $\min \{t(v), k\}$ links, we are done. If not, then $v$ cannot be activated by external influence alone, and so $v$ must have a neighbor $w$ that is activated before 
it, and that contributes to the activation of $v$. Let $T_1$ and $T_2$ be the two trees created by removing the edge between $v$ and $w$, with $T_1$ containing $w$, and let $S_1$ 
(respectively $S_2$) be the links of $S$ with an endpoint in $T_1$ (respectively $T_2$). Since $T$ is a tree, and $v$ is activated after $w$ by $S$, none of the links in $S_2$ can contribute to the activation of nodes in $T_1$. It follows that $S_1$ is a pervading link set for  $T_1$, and in fact is optimal, as a smaller solution for $T_1$ could be combined with $S_2$ to yield a better solution for $T$, contradicting the optimality of $S$.  By the inductive hypothesis, there is an optimal solution $S'$ for $T_1$ that gives $\min \{t(w), k\}$ links to $w$. Note that $|S'| = |S_1|$, and $S' \cup S_2$ must also be an optimal solution for $T$. 
Let $\extinf_i$ be an external influencer not giving a link to $v$ in $S' \cup S_2$, and let $\extinf_j$
be an external influencer giving a link to $w$ in $S' \cup S_2$ (note that $\extinf_i$ and $\extinf_j$ must exist).
Clearly $S''= S' \cup S_2  \cup \{ (\extinf_i, v) \} - \{ (\extinf_j, w) \}$ also activates the entire tree $T$ (because the $w$ influence on $v$ is replaced by $(\extinf_i, v)$, and
so $v$ still activates, and the $(\extinf_j, w)$ influence on $w$ is replaced by $v$'s activation).  Moreover since $|S''| = |S|$, we conclude that $S''$ is an optimal solution for  $T$.
But $S''$ gives more links to $v$ than $S$, contradicting our choice of $S$.  We deduce that there is an optimal pervading link set that gives $\min \{t(v), k\}$ links to $v$, as needed to complete the proof by induction. 
\end{proof}

The above lemma suggests a simple way to break up the {\sf Min-Links} problem for a tree into subproblems that can be solved independently, which yields a linear-time greedy algorithm. 

\begin{theorem}
The {\sf Min-Links} problem can be solved for  trees in linear time. 
\label{tree-algo}
\end{theorem}

\begin{proof}
Given a tree $T$, let $v$ be an arbitrary leaf in the tree. By Lemma~\ref{lemma:tree}, there is an optimal solution, say $S$,  to the {\sf Min-Links} problem for $T$ that gives $\min \{t(v), k\}$ links to $v$.
Let $S_v$ be the set of links given to $v$. 
Suppose $t(v) >k$, then the links given to $v$ are not enough to activate $v$, and therefore $v$'s neighbor $w$ must contribute to the activation of $v$. 
Also, $v$'s activation cannot help in activating any other nodes in $T$. Thus $S - S_v$ must be an optimal solution to $T' = T - \{v \}$.  Suppose instead that $t(v) \leq k$. Then the links given to $v$ activate it immediately. Consider the induced subgraph $T^{(1)}$ of $T$ that contains $v$, plus every node of $T$ of threshold 1. 
Let $C$ be the connected component (subtree) of $T^{(1)}$ that contains $v$ (note that $C$ might have only $v$). Then clearly $v \ac C$. Since $S$ is optimal, $S$ cannot contain any link to a node in $C$ except for $v$. Construct $T'$ by removing $C$ from $T$, and subtracting 1 from the threshold of any node $x$ who is a neighbor of a node in $C$. Observe that any such node $x$ can be a neighbor of exactly one node in $C$, since $T$ is a tree. Then $S - S_v$  must be an optimal solution to $T'$; if instead there is a smaller-sized solution to $T'$, we can add the links from $S_v$ to $v$ to that solution to obtain a smaller solution for $T$ than $S$, contradicting the optimality of $S$. 

The above argument justifies the correctness of the following simple greedy algorithm. Initialize $S = \emptyset$. Take a leaf $v$ in the tree. If $t(v) > k + 1$ then there is no solution by Observation~\ref{impossible}. If $t(v) = k + 1$, then give $k$ links to $v$ in $S$, 
remove $v$ from the tree, and recursively solve the remaining tree.  If $t(v) \leq k$, 
then give $t(v)$ links to $v$ (from arbitrary influencers), remove the subtree of $T$ that is connected to $v$ consisting only of nodes of degree 1, reduce the thresholds of all neighbors of the nodes in this subtree by 1, and recursively  solve the resulting trees. It is easy to see that the algorithm can be implemented in linear time. 
 \end{proof}


For the  network in Figure~\ref{link-set-example}, assuming that leaves in the tree are always processed in alphabetical order, the greedy algorithm given in Theorem~\ref{tree-algo} first picks node $b$ and adds a link to it.  We then remove nodes $b$ and $a$, and reduce the threshold of $d$ by 1. Next we pick $c$, give it a link, remove it from the tree, and decrement $t(f)$ to 2. The next leaf that is picked and given a link is $d$; since $d$'s threshold now is 1, we remove $d$ and $e$ from the tree, and reduce $f$'s threshold to 1. Proceeding in this way, we arrive at the link set shown.

We now give an exact bound on $ML(T)$, the number of links required to activate the entire tree $T$:

\begin{theorem}
Let $T$ be a tree that has a pervading link set. Then $ML(T) =   1 + \sum_{v \in T}  (t(v)-1) $	
\label{tree-ML}
\end{theorem}

\begin{proof}
We give a proof by induction on the number of nodes $n$ in the tree. Clearly if the tree consists of a single node $x$, there is a solution if and only if $t(x) \leq k$, and the number of links needed is $t(x)$ which is equal to $ 1 + \sum_{v \in V}  (t(v)-1)$ as needed. Now consider a tree $T$ with $n > 1$ nodes and let $x$ be a leaf in the tree. Then by Lemma~\ref{lemma:tree}, there is an optimal solution $S$ in which $x$ gets a set $S_x$ of $\min \{t(x), k\}$ links. By Observation~\ref{impossible}, there is a solution  only if $t(x) \leq k + 1$.  
Let $T'= T- \{x\}$ (all nodes keep the same thresholds as in $T$) and let $T''$ be the tree derived from $T$ by removing $x$ and reducing the threshold of $w$, the neighbor of $x$ in $T$ by 1.

First we consider the case when  $t(x) = k + 1$. Then giving $k$ links to $x$ from $S_x$ is not sufficient to activate it. 
By the usual cut-and-paste argument,   $S - S_x$ 
must be an optimal solution for tree $T'$.  

\vspace*{-0.1in}
\begin{eqnarray*}
ML(T) & = &k + ML(T')\\
& = & t(x) - 1 + ( 1+ \sum_{v \in V(T')} (t(v) - 1)) \mbox{ by the inductive hypothesis }\\
& = &  1+ \sum_{v \in V(T)} (t(v) - 1).
\end{eqnarray*}
\vspace*{-0.1in}

Next  we consider the case when $t(x) \leq k$, and  $t(w) > 1$. Then $x$ is immediately activated by the $t(x)$ links it receives in $S$, and the activation of $x$ effectively reduces  the threshold of $w$. Therefore, $S - S_x$ must be an optimal solution for the tree $T''$ in which the threshold of $w$ is $t(w) -1$. It follows that

\vspace*{-0.1in}
\begin{eqnarray*}
ML(T) & = &t(x) + ML(T'')\\
& = & t(x) + ( 1+ \sum_{v \in V(T'')} (t(v) - 1)) \mbox{ by the inductive hypothesis }\\
& = &  t(x) + 1 + (t(w) - 2) +  \sum_{v \in V(T'') - \{ w\}} (t(v) - 1) \\
& = & 1 + \sum_{v \in V(T)} (t(v) - 1).
\end{eqnarray*}
\vspace*{-0.1in}	

Finally suppose $t(x) \leq k$ and $t(w) = 1$.  
Then it is impossible that $S$ contains a link to $w$, as this would contradict the optimality of $S$. Therefore, we can move one link from node $v$ to node $w$, to get a new optimal pervading link set $S'$ for $T$. Furthermore, $S' - S_x$ must also be an optimal pervading link set for $T'$. It follows that 

\vspace*{-0.1in}
\begin{eqnarray*}
ML(T) & = & ML(T')\\
& = & t(x) - 1 + ( 1+ \sum_{v \in V(T')} (t(v) - 1)) \mbox{ by the inductive hypothesis }\\
& = &  1+ \sum_{v \in V(T)} (t(v) - 1).
\end{eqnarray*}
\vspace*{-0.1in}
\end{proof}

We remark that in contrast to the intuition for the optimal target set problem, where we would choose nodes of high degree or threshold to be in the target set, in the {\sf Min-Links} problem, our algorithm gives links to leaves initially, though eventually nodes that were internal nodes in the tree may also receive links. That is, the best nodes to befriend might be the nodes with a single connection to other nodes in the tree!

\subsection{Cycles}

In this section, we give a solution for the \ML problem on cycles. 
Let $C_n = (V, E, t)$ be a cycle with $n$ nodes, $V = \{0, 1, ... , n - 1\}$, $E 
= \{((i,  i + 1)\ mod\ n) \ |\ 1 \leq i \leq n \}$, 
and $t: t(v) \rightarrow \cal{Z^{+}}$.   
We define $P_{i, j}$ $ (i \neq j )$  to be the sub-path of $C_n$ consisting of 
all nodes in $\{i, \ldots, j\}$ in the clockwise direction.  We may use 
the $[i, j]$ notation to denote the vertices of $P_{i, j}$.
By \emph{consecutive vertices of threshold $k + 2$}, we mean two vertices $i, j$ such that 
the only two vertices in $P_{i ,j}$  with threshold $k + 2$ are $i$ and $j$.

\begin{proposition}
A  cycle has a pervading link set if and only if the following conditions hold:\\
(1) there is at least one node of threshold at most $k$,\\
(2) every node is of threshold at most $k + 2$,\\
(3) between any  two consecutive nodes of threshold $k + 2$, there is at least one node of  threshold at most $k$. 

\label{necc-suff-cycle}
\end{proposition}

\begin{proof}
The necessity of the first two conditions follows from Observation~\ref{impossible}. As for the third condition, suppose there are two consecutive nodes $i$ and $j$ of threshold $k + 2$, 
such that all nodes between them have threshold $k + 1$. Then both nodes $i$ and $j$  needs both their neighbors to be activated before them (even if they receive $k$ links), but meanwhile, since there is no node of threshold $k$ or less in $[i + 1, j - 1]$,   
no node in the sub-path $P_{i+1, j-1}$ can be activated. Therefore none of the nodes in the sub-path $P_{i, j}$ can be activated.  Conversely, if all three conditions listed in the statement are met, it is easy to see that by giving $k$ links to every node in the cycle, all the nodes in the cycle can be activated. 
\end{proof}

We note that a similar condition can be stated for paths, with the additional restriction that there must be a node of threshold at most $k$ before (after) the first (last resp.) node of threshold $k + 2$, if any. 

We give a linear time algorithm for finding a minimum-sized link set for 
problem {\sf Min-Links}$(C_n)$. Essentially we reduce the problem to finding an optimal solution for an appropriate path.

\begin{theorem}
	The {\sf Min-Links} problem for a cycle $C_n$ can be solved
	in linear time.
\label{cycle-algo}
\end{theorem}

\begin{proof}
If all nodes are of threshold 1, or if  there is a single node with threshold $2$, and the remaining nodes all have threshold 1, then  by giving 
a link to {\em any} of the nodes with threshold 1, we can activate the entire cycle, and this is clearly optimal. 

Therefore, in what follows, we assume that one of the following cases holds:

\begin{enumerate}[(a)]
\item the minimum threshold is greater than 1, 

\item the minimum threshold is 1,  and there are at least two nodes with threshold  $\geq 2$,  or 
 
 \item the minimum threshold is 1 and there is exactly one node with threshold $>2$.  
 
 \end{enumerate}
 
 Fix an arbitrary  node $i$ of 
minimum threshold  in $C_n$. We define $c(i)$ and $cc(i)$ to be the first 
node with threshold $>1$ in $i$'s clockwise direction and counter clockwise 
direction respectively. Observe that  in cases (a) and (b), $c(i) \neq cc(i)$ (see Figure \ref{fig:cycle1}), and  in case (c) above,  $c(i)=cc(i)$. We also define $P_{c(i), cc(i)}$ to be the path from $c(i)$ to $cc(i)$,   except that  in cases (a) and (b), we decrement $t(c(i))$ and $t(cc(i))$ by 1; and in case (c), the path contains a single node $c(i)$, and we decrement the threshold of $c(i)$ by 2. We now prove that  an optimal solution to {\sf MinLinks}($C_n$) can be constructed by giving $t(i)$ links to $i$ and combining them with an optimal solution to $P_{c(i), cc(i)}$. 

 We first claim that  there exists an optimal solution that gives  $t(i)$ links to $i$.  To see this, let $S$ be an optimal solution that gives $q < t(i)$ links to  node $i$. Observe that $q \in \{ t(i)-1, t(i)-2 \}$, as otherwise it is impossible for $i$ to be activated. First suppose $q = t(i)-1$. This means one of $i$'s neighbours must activate $i$. We follow the chain of activation to $i$, which must start at some node $j$ which is activated entirely by external influence. That is, $S$ must give $t(j)$ links to $j$. Without loss of generality, we assume $j \ac j+1  \ac \ldots \ac i-1 \ac i$. Let $\mu$ be such that
 $(\mu, j) \in S$ and  $(\mu , i) \notin S$. Such a $\mu$ must exist since $j$ received $t(j)$ links and $t(j)  \geq t(i)$ as $i$ was a node of minimum threshold in $C_n$ while $i$ received  $< t(i)$ links by assumption.  We now construct   a new solution $S' = S - \{ (\mu, j) \} \cup \{(\mu, i) \}$ by moving a link from $j$ to $i$. Since $i$ receives $t(i)$ links in $S'$, the node $i$ is immediately activated. Furthermore, $i \ac i-1 \ac \ldots \ac j+1$. Finally, since $j$ receives $t(j)-1$ links in $S'$, and $j+1$ is activated, node $j$ is activated in the next step. 
 Thus $S'$ is a solution of the same size as $S$, that gives $t(i)$ links to $i$ as needed. 
 
 Next suppose $q = t(i)-2$. Then both neighbours of $i$ must be activated by $S$ before $i$,  and both serve to activate $i$. We then follow the chain of activation in both directions from $i$ and find nodes $p$ and $q$ that were activated entirely by external influence (it is possible that $p=q$). We then move a link from each of $p$ and $q$ to $i$. Now $i$ is activated entirely by external influence and eventually activates $p$ and $q$. This completes the proof of the claim that there exists an optimal solution that gives $t(i)$ links to $i$.

 Consider therefore an optimal solution $S$ that gives $t(i)$ links to the node $i$; let us call this set of links $S_i$.   It is not hard to see that 
that $S - S_i$ must be an optimal solution to {\sf Min-Links}$(P_{c(i), cc(i)})$, since activating $i$ activates
$[cc(i) + 1, c(i) - 1]$ and lowers the threshold of $cc(i)$ and $c(i)$ by 1 each in cases (a) and (b) above and lowers the threshold of $c(i)$ by 2 in case (c) above. 

Finally, since the 
{\sf Min-Links} problem for a path can be solved in linear time according to Theorem 
\ref{tree-algo}, we can construct an optimal solution for a cycle in 
linear time as well.
\end{proof}
 
\begin{figure}[ht!]
  \centering
  \includegraphics[]{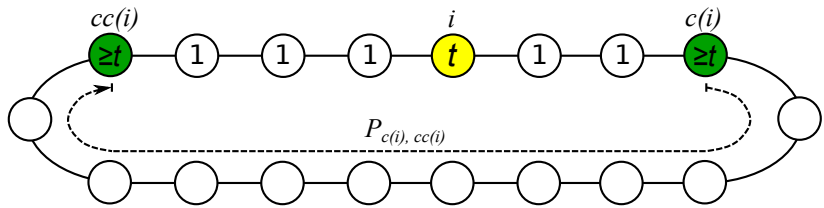}
  \caption{A cycle in which case (b) holds.  Node $i$ has minimum threshold $t$, while $cc(i)$ and $c(i)$ are the nodes closest to $i$ that have threshold higher than $1$.}
  \label{fig:cycle1}
\end{figure}

We give an exact bound on the number of links required to fully activate a cycle. 

\begin{theorem} 
	Given a cycle $C_n = (V, E, t)$ which has a pervading link set,  then
$$ML(C_n) = max \{ 1, \sum_{j=1}^n (t(j) - 1) \}.$$
\label{cycle-ML}
\end{theorem}
\begin{proof} 
If all nodes have threshold 1,  then $ML(C_n) = 1 = max \{ 1, \sum_{j=1}^n (t(j) - 1) \}$.  If one node has threshold 2, and all the remaining nodes have threshold 1, then
$ML(C_n) = 1 = \sum_{j=1}^n (t(j) - 1)$. 
Finally, for all remaining cases, it follows from the optimality of our algorithm that $ML(C_n) = t(i) + ML(P_{cc(i), c(i)})$ where $i$ is a node of minimum threshold in $C_n$, and the value  of $t(c(i)) + t(cc(i))$  is 2 less in $P_{cc(i), c(i)}$ than in $C_n$. By Theorem~\ref{tree-ML}, we have $ML(P_{cc(i), c(i)})= - 1+ \sum_{j \in [cc(i), c(i)]}(t(j) -1)$.
Therefore $ML(C_n) = t(i) - 1 + \sum_{j \in [cc(i), c(i)]}(t(j) -1) = \sum_{j=1}^n (t(j) - 1)$. 
 \end{proof}

\subsection{Cliques}

In this section, we give an algorithm to solve the \ML problem on cliques. 
Let $K_n=(V, E, t)$  be a clique with $n$ nodes, $V = \{1, 2, ... , n\}$ and 
$E = \{(i,  j) :  1\leq i <  j \leq n  \}$ and $t: t(v) \rightarrow 
\cal{Z^{+}}$.  We first show a necessary and 
sufficient condition for the \ML problem to have a feasible solution:

\begin{proposition} 
Let $K_n$ be a clique with $t(i) \leq  
t(i+1)$, for all $1 \leq i < n$.  Then $K_n$ has a pervading link set if and only if 
$t(i) \leq i + k - 1$ for all $1 \leq i \leq n$.	
\label{necc-suff-clique}
\end{proposition}
\begin{proof}  If $t(i) \leq k + i - 1$ for all $1 \leq i 
\leq n$, it is easy to see that there exists a solution $S$ by giving $k$ links to 
every node $i$; we claim 
that node $i$ is activated in or before round $i$. Since $t(1) \leq k$, node 1 is activated in round 1.  Inductively, node 1 to $i-1$ 
are already activated in round $i-1$, the effective threshold of node $i$ has 
been reduced to at most $k$. Node $i$ receives $k$ links, therefore, node $i$ must be 
activated in the $i^{th}$ round, if it is not already activated.
Conversely, suppose there exist nodes $j$ such that $t(j) > k + j - 1$ and there exists a solution $S$  to the {\sf Min-Links} problem; let $p$ be  the smallest such node with  $t(p) > k + p - 1$. In order to activate any node $q \geq p$, at least $p$ nodes have to be activated before $q$, since $t(q) \geq t(p) > k + p - 1$. 
However, there are only $p-1$   nodes that can be activated before any such node $q \geq p$. Thus no node $q$ with $q \geq p$ can be activated, a contradiction.
\end{proof}

We now give a greedy algorithm to solve the {\sf Min-Links} problem on a clique.
 
\begin{theorem}
	The \ML problem for a clique  $K_n$ can be solved in time
	$\Theta(n + k)$. 
\label{clique-algo}
\end{theorem} 

\begin{proof}
First sort the nodes in order of threshold. By Observation~\ref{impossible}, there is no solution if any node has a threshold greater than $n + k$, therefore, we can use counting sort and complete the sorting in $\Theta(n + k)$ time. 
Clearly, the 
condition given in Proposition~\ref{necc-suff-clique} can easily be checked in 
linear time. We now give the following greedy linear time algorithm for a clique which has a feasible solution: give $t(1)$ links to node 1 in order to activate it, and let $j > 1$ be the smallest value such that
node $j$ is not activated after activating node $1$.  Remove all nodes in $\{1, \ldots, j-1\}$, decrement by $j-1$ the thresholds of all nodes $\geq j$, and solve the resulting graph recursively.  It is easy to see that this algorithm can be implemented in linear time,  in an iterative fashion as follows: we examine the nodes in order. When we process node $i$, if $t(i) < i$, we 
simply increment $i$ and continue; if $t(i) \geq i$,  we give $t(i) - i + 1$ links to node $i$
(note that we assume $t(i) \leq k + i - 1$, and so $t(i) - i + 1 \leq k$). We now show that the link set produced by this greedy algorithm is 
optimal. 

First we show that there must 
be an optimal solution that gives $t(1)$ links to the  node $1$. Consider an optimal solution
$S$ in which node $1$ gets $q < t(1)$ links.  We follow the chain of activations by $S$ to node $1$. Let $A_1$ be the set of nodes that is activated by external influence alone according to solution $S$, and suppose $A_1 \ac A_2 \ac \ldots  \ac A_j \ac 1$. Then $|A_1 \cup A_2 \cup \cdots \cup  A_j| \geq t(1) - q$. We construct a new solution $S'$ by moving
$t(1) -q$ links from some node $i \in A_1$ to node 1. Observe that $i$ received $t(i) \geq t(1)$ links in $S$, and so there are enough links to move. Furthermore, the nodes in $A_1' = A_1 - \{i \} \cup \{1\}$ are activated entirely by external influence in $S'$. Also  since $|A_1'| = |A_1|$, and all nodes are connected, $A_1' \ac A_2 \ldots \ac A_j$. Finally, together with the $t(i) - t(1) + q$ links that $S'$ gives to $i$, we have $A_1 \cup A_2 \cup A_j \ac i$. The rest of the activation proceeds in the same way in $S$ and $S'$. Since $S'$ is a solution of the same size as $S$ that gives $t(1)$ links to node $1$, the claim is proved.

Let $S_1 \subseteq S$ be the links given to node $1$ in $S$. Next we claim that 
$S - S_1$ is an optimal solution to the clique $C'$ which is the induced 
sub-graph on the nodes $\{j, j+1, \dots, n\}$ where $j>1$ is the smallest index 
with $t(j) \geq j$, and with thresholds of all nodes reduced by $j-1$.  Suppose 
there is a smaller solution $S'$ to $C'$. We claim that $S' \cup S_1$ 
activates all nodes in the clique $K_n$.  Since for any node $1 < k < j$, we have $t(k) < k$, it can be seen inductively that $S_1$  suffices to activate node $k$. Thus, all nodes in 
$\{1, 2, \ldots j-1 \}$ are activated. Furthermore, the thresholds of all nodes in $\{j, j+1, \dots, n\}$ 
are effectively reduced by 
$j-1$. Thus using the links in $S'$ suffices to activate them. Finally, since 
$|S'| < |S - S_1|$, $S' \cup S_1$ is a smaller solution than $S$ to the clique 
$K_n$, contradicting the optimality of $S$ for $K_n$. We conclude that the greedy algorithm described above 
produces a minimum sized solution to the {\sf Min-Links} problem.
\end{proof}

The following tight bound on the minimum number of links to activate an entire clique is immediate:

\begin{theorem}
Given a clique $K_n$ which has a feasible solution, let $P = \{i : t(i) \geq i\}$.
Then 
$$ML(K_n) = \sum_{i \in P} t(i) - i + 1.$$
\label{clique-ML}
\end{theorem}

\section{An algorithm for general graphs} \label{sec:TPI}
In this section we design an algorithm, that works for arbitrary graphs $G=(V,E)$,  to efficiently allocate
links to nodes in $V$ from a set of external nodes $U$ in such a way that, assuming that the nodes in $U$ are already activated, they trigger
an influence diffusion process that activates the whole network.
We assume $|U|\geq t_{max}$, where $t_{max}$ is the maximum threshold of nodes in $V$ so that there is always an feasible solution for the {\sf Min-Links} problem.  
Our procedure is formally presented in Algorithm \ref{alg2}.

The algorithm works by computing a link vector
$\bs = (s(v_1), \ldots , s(v_n))$, where 
$s(v)$ is an integer representing  the number of links between  external nodes in $U$ and  the vertex $v \in V$.
As observed in Section \ref{notation}, the $s(v)$ links to $v$ can be seen as a partial incentive to node $v$.
From any  link vector $\bs$,  one  can get  link sets $S$ between nodes in $U$ and nodes in $V$. 
In the following we say that a link vector $\bs$ is a pervading link 
vector when any   corresponding set $S$ is a {\em pervading link set}.

The algorithm proceeds by iteratively deleting nodes from the graph $G$, and at each iteration the node to be deleted is chosen to maximize a certain parameter (Case 2). 
If, during the deletion process, a node $v$ in the surviving graph  remains with less neighbors 
than its current threshold   (Case 1), then a set of links (or equivalently a partial incentive) is added to $v$ so that $v$'s new threshold is equal to the number of neighbours of $v$ 
in the surviving graph.

\medskip

In the sequel,  we denote by 
 $\Active(G,\bs,j)$  the set of  nodes that are active at step $j$ of the influence 
diffusion process on the network $G$ augmented with a set of links determined  by the link vector $\bs$. 
Namely, let $\Gamma_G(v)$   denote the neighborhood of
 $v $ in $G$, we have
$\Active(G,\bs,0) = \{v \ |\  s(v)\geq t(v)\}$ and 
$\Active(G,\bs,j)=\Active(G,\bs,j-1)\cup \{u\ |\ |\Gamma_G(v)\cap \Active(G,\bs,j-1)\geq t(u)-s(u)\}$, for all $j\geq 1$.
 The pseudocode for our algorithm is given in Algorithm TPI($G$).

\begin{figure}
\begin{center}
\includegraphics[width=8truecm]{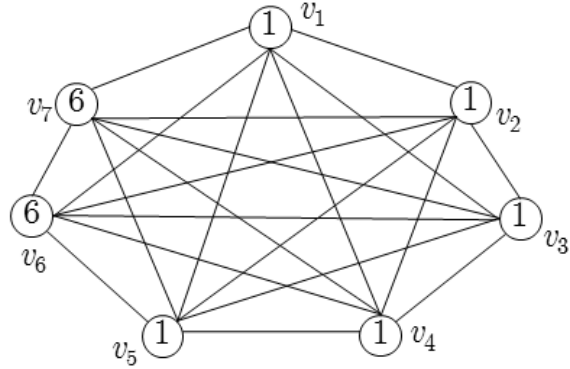}
\caption{An example to illustrate Algorithm TPI($G)$. The number inside each circle is the  node threshold.\label{fig:k}}
\end{center}
\end{figure}
\def\W {W}
\begin{algorithm}
\SetKwInput{KwData}{Input}
\SetKwInput{KwResult}{Output}
\DontPrintSemicolon
\caption{ \ \   \textbf{Algorithm} TPI($G$) \label{alg2}}
\KwData { A graph $G=(V,E,t)$ where $t$ is the threshold function on $V$. }
\KwResult{ $\bs = (s(v_1), \ldots  , s(v_n))$ a link vector for $G$, where $V = \{v_1, \ldots , v_n\}$.}
\setcounter{AlgoLine}{0}
$\W=V$; \\  
\ForEach {$v \in V $}{
	$\cs( v)=0;$  \tcp*[f]{\# of links  to $v$.} \\
  $\d( v)=d(v)$;\\ 
  $k( v)=t( v)$; \\
  $N( v)=\Gamma_G( v)$; 
}
\While(\tcp*[f]{Select one node and either update its links or remove it from the graph.}){ $\W\neq \emptyset$ }{
\eIf(\tcp*[f]{\underline{Case 1}: Increase $s(v)$ and update $k(v)$}.){ $\exists\ v\in \W$ s.t. $k(v)>\d(v)$}
		{
		$\cs(v)=\cs(v)+ k(v)-\d(v)$;\\
		$k(v)=\d(v)$;\\
		\If(\tcp*[f]{here\  $\d(v)=0$.}){k(v)=0}{$\W=\W-\{v\};$}
		}
		(\tcp*[f]{\underline{Case 2}: Choose a node $v$ to eliminate from the graph. })	
				{
    		$v={\tt argmax}_{u\in \W}\left\{\frac{k(u)(k(u)+1)}{\delta(u)(\delta(u)+1)}\right\}$;\\
				\ForEach{$u\in N(v)$}{
    						$\d(u)=\d(u)-1$;\\
								$N(u)=N(u)-\{v\};$
    				}
						$\W=\W-\{v\};$
			  }	
				\KwRet $\bs$
}
\end{algorithm}	
\begin{example}\label{ex-4a}
Consider a complete graph on 7 nodes with thresholds $t(v_1)=\ldots= t(v_{5})=1$,   $t(v_{6})=t(v_{7})=6$ (cf. Fig. \ref{fig:k}).
A possible execution of the algorithm is summarized below. 
At each iteration of the while loop, the algorithm   considers  the nodes  in the  order
shown in the table  below, where  we also 
 indicate for each node whether  Cases 1 or 2 applies and the updated value of the 
number of links for the selected node:

\begin{center}
\begin{tabular}{|l|l|l|l|l|l|l|l|l|}
\hline 
Iteration      & $1$    & $2$      & $3$      & $4$    & $5$    & $6$    & $7$    & $8$   \\
\hline 
Node      & $v_{7}$       & $v_{6}$     & $v_{6}$    & $v_1$     & $v_2$        & $v_3$          & $v_4$     & $v_5$ \\
\hline 
Case      & $2$    & $1$        & $2$      & $2$         & $2$      & $2$       & $2$    & $1$\\
\hline 
 Links    & $0$     & $1$       & $1$   
                         & $0$      & $0$     & $0$       &  $0$     &  $1$\\
 \hline 
\end{tabular}
\end{center}

The algorithm  $TPI(G)$  outputs  the  link vector $\bs$ having non zero elements
$s(v_5)=s(v_6)=1$, for which we have 
\begin{eqnarray*}
 \Active(G,\bs,0) &=& \{v_5\}  \qquad \mbox{\em (since $s(v_5)=1=t(v_5)$)} 
 \\
 \Active(G,\bs,1)&=&\Active(G,\bs,0)\cup\{v_1,v_2,v_3,v_4\}=\{v_1,v_2,v_3,v_4,v_5\}
 \\
 \Active(G,\bs,2)&=&\Active(G,\bs,1)\cup \{v_6   \} =\{v_1,v_2,v_3,v_4,v_5, v_6   \} \qquad \mbox{\em (since $s(v_6)=1$)}  
\\
 \Active(G,\bs,3)&=&\Active(G,\bs,2)\cup \{  v_7 \}   =V.
\end{eqnarray*}
\end{example}
\vspace*{-.7truecm}

We first prove the correctness and complexity of the algorithm.  Subsequently, we compute an  upper bound on the number of links 
in the pervarding link set,  that is, an upper bound on
$\sum_{v\in V}s(v)$.
To this end, we define the following notation. 

Let $\ell$ be the number of iterations of the while loop in TPI($G$).
For each iteration $j$ of the while loop, with $1 \leq j \leq \ell$,   we denote
\begin{itemize}
\item by $\W_j$ the set $\W$ 
at the beginning of the $j$-th iteration (cf. line 7 of $TPI(G)$), in particular $\W_1=V(G)$ and $\W_{\ell+1}=\emptyset$;
\item by $\G(j)$  the subgraph of $G$ induced by the vertices in $\W_j$;
\item by $\vj$ the node selected during the $j$-th iteration\footnote{A node can be  selected  several  times before being eliminated; indeed in Case 1  we can have $\W_{j+1}=\W_j$.};
\item by $\d_j(v)$  the degree of node $v$ in $\G(j)$;
\item by  $k_j(v)$  the value of the threshold of node $v$ in $\G(j)$, that is, as it is updated at 
the beginning of the $j$-th iteration, in particular $k_1(v)=t(v)$ for each $v\in V$;
\item by  $s_j(v)$  the number of links to  $v$  that are computed by the algorithm from  the  $j$-th  iteration until  and including the $\ell$-th iteration, in particular observe that   $s_{\ell+1}(v)=0$ and   $s_{1}(v)=s(v)$ for each $v\in V$;

\item by $\s_j$   the number of links assigned  {\em during} the $j$-th iteration, that  is,
$$\s_j=\cs_j(v_j)-\cs_{j+1}(v_j)=\begin{cases} {0}&{\mbox{if} \ k_j(\vj) \leq \d_j(\vj),}\\
                      {k_j(\vj) - \d_j(\vj)}&{\mbox{otherwise.}} 
				\end{cases}$$
\end{itemize}
According to  the above notation, we have that if node $v$ is  selected during the iterations
$j_1 < j_2 < \ldots < j_{a-1} < j_a$ of the while loop in TPI($G$),  where  the last value  $j_a$  is the iteration when  $v$ 
has been eliminated from the graph, then 
$$\cs_j(v)=\begin{cases}  {\s_{j_1}+\s_{j_2}+\ldots   + \s_{j_a}}&{\mbox{if} \  j \leq j_1,}\\
               {\s_{j_b} +  \s_{j_{b+1}}+\ldots  + \s_{j_a} } &{\mbox{if} \ j_{b-1} < j \leq j_b {\leq j_a},}\\
                       {0}&{\mbox{if} \ j > j_a.}
				\end{cases}$$
{In particular observe that for $j=j_a$, we have $\cs_j(v_j)=\s_j$.}		

The following result is immediate.											
\begin{prop}\label{prop1}
Consider  the node $\vj$ that is  selected during   iteration $j$, for $1 \leq j \leq \ell$, of the while loop in the algorithm TPI($G$):
\begin{itemize}
\item[ {\bf (1.1)}] If Case 1 of the algorithm 
TPI($G$) holds and  $\d_{j}(\vj) =0$, then 
 $k_{j}(\vj)> \d_{j}(\vj) =0$ and  the isolated node $\vj$ is eliminated from $\G(j)$.
Moreover,\\  
$$\W_{j+1}=\W_{j} - \{\vj\},   \quad
\cs_{j+1}(\vj)= \cs_{j}(\vj)-\s_j,   \quad
 \s_{j}= k_{j}(\vj)-\d_{j}(\vj)=k_{j}(\vj)>0, $$ 
and, for each  $v {\in} \W_{j+1}$
$$\cs_{j+1}(v)=\cs_{j}(v), \quad
 \d_{j+1}(v)= \d_{j}(v), \quad 
 k_{j+1}(v)= k_{j}(v).$$
\item[ {\bf  (1.2)}] If Case 1 of  TPI($G$) holds with  $\d_{j}(\vj) >0$, then 
 $k_{j}(\vj)> \d_{j}(\vj) >0$ and  no node is deleted from $\G(j)$, that is, $\W_{j+1}=\W_{j}$.
 Moreover, 
$$\s_{j}= k_{j}(\vj)-\d_{j}(\vj)>0$$ 
and, for each $v \in \W_{j+1}$
$$\d_{j+1}(v)= \d_{j}(v),\quad
\cs_{j+1}(v)=\begin{cases}  { \cs_{j}(\vj){-}\s_j }&{\mbox{if }  v=\vj}\\
													 {     \cs_{j}(v)        }&{\mbox{if } v \neq \vj }
													\end{cases}, \
\quad
 k_{j+1}(v)=\begin{cases}  {  \d_{j}(v)}&{\mbox{if} \ v=\vj}\\
													 {  k_{j}(v)   }&{\mbox{if} \  v \neq \vj}.
													\end{cases}												$$
\item[{\bf  (2)}] If Case 2 of TPI($G$) holds then 
 $k_{j}(\vj) \leq \d_{j}(\vj)$ and 
$\vj$ is pruned from $\G(j)$. Hence, 
$$\W_{j+1}=\W_{j} - \{\vj\}, \qquad \s_{j}=0,$$
  and,  for each $v \in \W_{j+1}$ it holds
$$
\cs_{j+1}(v)= \cs_{j}(v),  \quad  k_{j+1}(v)= k_{j}(v)		\quad 
\d_{j+1}(v)=\begin{cases}  {  \d_{j}(v) -1}&{\mbox{if} \ v \in \Gamma_{\G(j)}(\vj)}\\
													 { \d_{j}(v)   }&{\mbox{otherwise.} }
													\end{cases}
$$					
\end{itemize}
\end{prop}

\begin{lemma} \label{d=0}
For each iteration $j=1,2,\ldots,\ell$, of the while loop in the algorithm TPI($G$), 
\begin{itemize}
\item[(1)] \ if  $\ k_j(\vj) > \d_j(\vj)$ then $\s_j= k_{j}(\vj)-\d_{j}(\vj)=1$;
\item[(2)] \ if $\d_j(\vj)=0$ then $s_j(\vj)=k_j(\vj)$.
\end{itemize}
\end{lemma}
\noindent
{\proof
First, we prove (1).
At the beginning of the algorithm, $t(u)=k(u)\le d(u)=\d(u)$ for all  $u\in V$. 
Afterwards, the value of $\d(u)$ is decreased by at most one unit for each iteration 
(cf. line 16 of \textit{TPI}($G$)). Moreover, 
the first time the condition of Case 1 holds for some node $u$, one has 
$\d_j(u)=k_j(u)-1$. Hence, if the selected node is $v_j=u$ then (1) holds; otherwise, some $v_j \neq u$, satisfying the condition of Case 1 is selected  and $\d_{j+1}(u)=\d_{j}(u)$ and $k_{j+1}(u)=k_{j}(u)$ hold.
Hence, when at some subsequent iteration $j'>j$ the algorithm selects $v_{j'} = u$, 
we have $\d_{j'}(u)=k_{j'}(u)-1$.
\\
To show  (2), it is sufficient to notice that at the iteration $j$ when node $v_j$ 
is eliminated from the graph,  $\cs_j(v)=\s_j$.
}\qed
\smallskip

We are now ready to prove the correctness and complexity of the algorithm TPI($G$):

\begin{theorem}\label{teo-att-par}
For any  graph $G=(V,E,t)$ and a set $U$ of external influencers such that $|U| \geq max_{v \in V} t(v)$,  the algorithm TPI($G$) outputs a pervading link vector for $G$ in time $O(|E|\log |V|)$.
\end{theorem}
\proof
We  show that for  each iteration $j$, with $1 \leq j \leq \ell$, the assignment of
$s_j(v)$  links for each $v \in \W_j$ activates all the nodes of the graph $\G(j)$ when the distribution of  thresholds to its nodes is $k_j(\cdot)$.
The proof is by induction on $j$.

\smallskip\noindent
If $j=\ell$ then  the unique node $v_\ell$ in $\G(\ell)$  has  degree $\d_\ell(v_\ell)=0$  and $s_\ell(v_\ell)=k_\ell(v_\ell)=1$ (see Lemma \ref{d=0}).

\smallskip\noindent
Consider now $j< \ell$ and suppose the algorithm is correct on $\G(j+1)$ that is,  the assignment of $s_{j+1}(v)$ links to each $v \in \W_{j+1}$, activates all the nodes of the graph $\G({j+1})$  when the distribution of thresholds to its nodes is $k_{j+1}(\cdot)$.  \\
Recall that $v_{j}$ denotes  the node the algorithm selects from  $\W_j$ (thus  obtaining $\W_{j+1}$,  the node set of $\G(j+1)$).
In order to prove the theorem we analyze  three cases according to the current degree and threshold of the selected node $v_{j}$.
\begin{itemize}
\item $k_{j}(\vj)> \d_{j}(\vj) =0$:  By Lemma \ref{d=0},  we have $k_j(\vj)=s_j(\vj)$. Then the correctness of the algorithm on $\G(j)$ follows from part (1.1) of Proposition \ref{prop1} and  the inductive hypothesis on $\G(j+1)$.
\item $k_{j}(\vj)> \d_{j}(\vj) \geq 1$:  From (1.2) of Proposition \ref{prop1},  we observe that
 $k_j(v)- s_j(v)=k_{j+1}(v)-s_{j+1}(v)$, for each node $v \in \W_j$. 
{Indeed, for each $v \neq v_j$ we have $k_{j+1}(v)=k_j(v)$ and $s_{j+1}(v)=s_{j}(v)$. 
Moreover, $$k_{j+1}(\vj)-s_{j+1}(\vj)= \d_j(\vj)-s_{j}(\vj)+\s_j=k_{j}(\vj)-s_{j}(\vj).$$}
   Hence the
nodes that can be activated in  $\G(j+1)$  can be activated in  $\G(j)$ with thresholds 
$k_{j}(\cdot)$ and $s_{j}(\cdot)$ links. 
So, the correctness of the algorithm on  $\G(j)$ follows from the inductive hypothesis on $\G(j+1)$.
\item $k_{j}(\vj) \leq \d_{j}(\vj)$: From part  (2) of Proposition \ref{prop1}, and by  the inductive hypothesis on  $\G(j+1)$, we have that all
the neighbors of $\vj$ in $\G(j)$ that are nodes in $\W_{j+1}$ gets activated; 
since $k_j(\vj) \leq \d_j(\vj)$, we conclude that $\vj$ also gets activated in $\G(j)$.
\end{itemize}

This completes the proof by induction. To prove the time complexity, we first  notice that the algorithm ends within $|E|$ iterations of the {\em while} loop.
Indeed,  each time  a vertex $v$ is selected in Case 1 of the algorithm, 
its current  threshold   $k(v)$ is decreased to the current node degree $\delta(v)$. 
Moreover,    if $k(v)$ reaches 0  then $v$ is deleted.
This implies that  each vertex v  can be selected up to $d(v)$ times  before being deleted
 (because its threshold gets down to  0 or because Case 2 is applied  to $v$).
Hence,  the algorithm  executes  at most $|E|$ iterations of the while loop.

When Case 2 applies,  one has to select the node  $u$ in the current node set
that maximizes   the quantity  $\frac{k(u)(k(u)+1)}{\delta(u)(\delta(u) + 1)}$. 
Suppose that  the nodes   are initially sorted in a max-heap 
according to the priorites  $\frac{k(u)(k(u)+1)}{\delta(u)(\delta(u) + 1)}$. 
Each time the node priority  $\frac{k(u)(k(u)+1)}{\delta(u)(\delta(u) + 1)}$ changes, we can  update the 
heap in $O(\log |V|)$ time.
For each iteration, if  $v$ is  the processed node then  the 
update involves  only the  neighbors of $v$ if 
Case 2 applies, and only the node  $v$  itself if Case 1 applies.
Overall,  Case 1 requires  at most $|E|\cdot O(\log |V|)$  time and   Case 2 
requires  at most $\sum_{v\in V} d(v) \cdot O(\log |V|)$ times.

Hence,  the  algorithm can be implemented in such a way to run
in  $O(|E|\log |V|)$ time.
\qed

\medskip
\noindent

We now compute an upper bound on the number of links that the link vector returned by the algorithm TPI assigns to the vertices in $V$.
\begin{theorem}\label{teo3}
For any  graph  $G$ the algorithm TPI($G$) returns a link vector $\bs$ for $G$ such that
$$\sum_{v\in V}\cs(v)\leq \sum_{v\in V}  \frac{t(v)(t(v) +1)}{2(d_G(v)+1)}.$$
\end{theorem}
\proof
Define  $\D(j)=\sum_{v \in \W_j} \frac{k_j(v)(k_j(v)+1)}{2(\d_j(v)+1)}$,
for each $j=1,\ldots, \ell$. 
By definition of $\ell$, we have $\G(\ell+1)$ is the empty graph;  we then define $\D(\ell+1)= 0$.
We prove now by induction on $j$ that  
\begin{equation}\label{S}
\s_j \leq \D(j) - \D(j+1).
\end{equation}
By using (\ref{S}) we will have the bound on $\sum_{v \in V}\cs(v)$. Indeed, 
$$\sum_{v \in V}\cs(v) = \sum_{j=1}^{\ell}\s_j \leq \sum_{j=1}^{\ell} (\D(j) - \D(j+1)) 
=\D(1)-\D(\ell+1) = \D(1) =  \sum_{v\in V}  
\frac{t(v)(t(v) +1)}{2(d(v)+1)}.$$
\\
In order to prove (\ref{S}), we analyze  three cases depending on {the relation between $k_{j}(\vj)$ and $\d_{j}(\vj)$}.\\
$\bullet$ Assume first $k_{j}(\vj)> \d_{j}(\vj) =0$. We get
\begin{eqnarray*}
\D(j)- \D(j{+}1) &=& \sum_{v\in \W_j} \frac{k_j(v)(k_j(v)+1)}{2(\delta_j(v)+1)}
     -\sum_{v\in \W_{j+1}} \frac{k_{j+1}(v)(k_{j+1}(v)+1)}{2(\delta_{j+1}(v)+1)}
\nonumber\\
&=& \frac{k_j(\vj)(k_j(\vj)+1)}{2(\delta_j(v_j)+1)} + \sum_{v\in \W_j -\{v_j\}} \frac{k_j(v)(k_j(v)+1)}{2(\delta_j(v)+1)}\\
     && \hphantom{aaaaaaaaaaaaaaa} - \sum_{v\in \W_{j+1}} \frac{k_{j+1}(v)(k_{j+1}(v)+1)}{2(\delta_{j+1}(v)+1)}
\nonumber\\
&=& 
\frac{k_j(\vj)(k_j(\vj)+1)}{2(\delta_j(v_j)+1)}  \qquad \qquad\mbox{{(by 1.1 in Proposition \ref{prop1})}}\\ 
&=& 1 = \s_j. \qquad \qquad\qquad \qquad\qquad\mbox{{(by Lemma \ref{d=0})}}
\end{eqnarray*}

\noindent
$\bullet$ Let now $k_{j}(\vj)> \d_{j}(\vj) \geq 1$. We have
\begin{eqnarray*}
\D(j)- \D(j{+}1)&=& \sum_{v\in \W_j} \frac{k_j(v)(k_j(v)+1)}{2(\delta_j(v)+1)}
     -\sum_{v\in \W_{j+1}} \frac{k_{j+1}(v)(k_{j+1}(v)+1)}{2(\delta_{j+1}(v)+1)}
\nonumber\\
&=& \frac{k_j(\vj)(k_j(\vj){+}1)}{2(\delta_j(v_j)+1)} -\frac{k_{j+1}(\vj)(k_{j+1}(\vj){+}1)}{2(\delta_{j+1}(v_j)+1)}
\nonumber\\
&& +  \sum_{v\in \W_j -\{v_j\}} \frac{k_j(v)(k_j(v)+1)}{2(\delta_j(v)+1)}
     -\sum_{v\in \W_{j+1} -\{v_j\}} \frac{(k_{j+1}(v)(k_{j+1}(v)+1)}{2(\delta_{j+1}(v)+1)}
\nonumber\\
&=& 
\frac{(\delta_j(v_j)+1)(\delta_j(v_j)+2)}{2(\delta_j(v_j)+1)} - \frac{\delta_j(v_j)(\delta_j(v_j)+1)}{2(\delta_j(v_j)+1)}
\nonumber \ \  \mbox{{(by 1.2 in Proposition \ref{prop1})}}\\
&=& 
\frac{2(\delta_j(v_j)+1)}{2(\delta_j(v_j)+1)} = 1 = \s_j.  \qquad \qquad\qquad \qquad\qquad\mbox{{(by Lemma \ref{d=0})}}
\end{eqnarray*}

\noindent
$\bullet$ Finally, let $k_{j}(\vj)  \leq \d_{j}(\vj)$. 
In this case, by the algorithm  we know that   

\begin{equation} \label{argmaxEq}
\frac{k_j(v)(k_j(v)+1)}{\delta_j(v)(\delta_j(v)+1)} \leq \frac{k_j(\vj)(k_j(\vj)+1)}{\delta_j(\vj)(\delta_j(\vj)+1)},
\end{equation}
for each $v\in  \W_j$.
Hence, we get
\begin{eqnarray*}
\D(j)-\D(j+1)  &=& \sum_{v\in \W_j} \frac{k_j(v)(k_j(v)+1)}{2(\delta_j(v)+1)}
     -\sum_{v\in \W_{j+1}} \frac{k_{j+1}(v)(k_{j+1}(v)+1)}{2(\delta_{j+1}(v)+1)} \\
&=& \frac{k_j(\vj)(k_j(\vj){+}1)}{2(\delta_j(v_j){+}1)} 
+  \sum_{v\in \Gamma_{\G(j)}(v_j)} \frac{k_j(v)(k_j(v){+}1)}{2(\delta_j(v){+}1)}\\
&& \hphantom{a}
     -\sum_{v\in \Gamma_{\G(j)}(v_j) } \frac{k_{j+1}(v)(k_{j+1}(v){+}1)}{2(\delta_{j+1}(v){+}1)} 
		               \qquad\qquad \mbox{(by 2 in Proposition \ref{prop1})}
\\
&=& 
 \frac{k_j(\vj)(k_j(\vj){+}1)}{2(\delta_j(v_j){+}1)} 
+  \sum_{v\in \Gamma_{\G(j)}(v_j)} \frac{ k_j(v)(k_j(v){+}1)}{2}\left( \frac{1}{(\delta_j(v)+1)}
     - \frac{1}{\delta_{j}(v)} \right)
\\
&=& 
 \frac{k_j(\vj)(k_j(\vj)+1)}{2(\delta_j(v_j)+1)} 
-  \sum_{v\in \Gamma_{\G(j)}(v_j)} \frac{k_j(v)(k_j(v)+1)}{2\delta_{j}(v)(\delta_j(v)+1)}
\\
&\geq & 
          \frac{k_j(\vj)(k_j(\vj)+1)}{2(\delta_j(v_j)+1)} 
-   \frac{k_j(v_j)(k_j(v_j)+1)\delta_{j}(v_j)}{2\delta_{j}(v_j)(\delta_j(v_j)+1)} \ \ \qquad	\mbox{(by (\ref{argmaxEq}))}\\ 
&=& 0 = \s_j
\qquad\qquad\qquad\qquad\qquad\qquad\qquad\qquad\qquad\qquad\qquad\qquad \Box
\end{eqnarray*}

\medskip

Experimental data showing the effectiveness of   algorithm TPI$(G)$ on real-life networks are given in \cite{Sirocco15},
 as well as proofs of its optimality in case of cliques and trees. However, the  time complexity does not becomes linear as for the algorithms given in Section \ref{Sec4}.

\section{Discussion}

In this paper, we introduced and studied the {\sf Min-Links} problem: given a social network $G$ where every node $v$ has a threshold $t(v)$ to be activated,  which minimum-sized set of  nodes should an already activated set of external influencers 
$S$ befriend, so as to influence the entire network? We showed that the problem  cannot be approximated to within a ratio 
of  $O(2^{\log^{1-\epsilon} n})$,  for any fixed $\epsilon>0$, 
unless $NP\subseteq DTIME(n^{polylog(n)})$. 
 In contrast, we gave exact  linear time algorithms that  solve the problem in  trees, 
cycles, and cliques, for any given set of $k$ external influencers. We also  
 gave an exact bound (as a function of the thresholds) on the number of links needed for such graphs. 
Moreover, we gave a polynomial time  algorithm that solves the problem 
 in  general graphs and derived an upper bound on the number of links used by the algorithm.
It would be interesting to generalize  these algorithms to find the minimum number of links required to influence a specified fraction of the nodes. Other directions include studying the case with non-uniform weights on the edges. Clearly, the problem remains NP-complete in general, but the complexity for special classes of graphs remains open. Another interesting question is that of maximizing the number of activated nodes, given a fixed budget of $\ell$ links.

\end{document}